\documentclass[aps,prd,onecolumn,11pt]{revtex4}

\usepackage{amsmath,amssymb,amsthm}
\usepackage{graphicx}
\usepackage{bm}
\usepackage{hyperref}
\usepackage{enumitem}
\usepackage{booktabs}
\usepackage{subcaption} 
\usepackage{multirow}
\usepackage{tabularx}
\usepackage{float}
\usepackage{tikz}
\usepackage{natbib} 

\newtheorem{thm}{Theorem}[section]
\newtheorem{corollary}[thm]{Corollary}

\newcommand{\be}{\begin{equation}}
\newcommand{\ee}{\end{equation}}
\newcommand{\bea}{\begin{eqnarray}}
\newcommand{\eea}{\end{eqnarray}}

\setlist[enumerate,1]{label=(A\arabic*),ref=A\arabic*}

\begin{document}

\title{Averaging Theory and Dynamical Systems in Cosmology:\\
A Qualitative Study of Oscillatory Scalar-Field Models}

\author{Genly Leon$^{1,2}$}
\author{Claudio Michea$^{3}$}
\affiliation{$^{1}$Departamento de Matem\'aticas, Universidad Cat\'olica del Norte, Avda. Angamos 0610, Casilla 1280, Antofagasta, Chile, Antofagasta, Chile}
\affiliation{$^{2}$Institute of Systems Science, Durban University of
Technology, Durban 4000, South Africa}
\affiliation{$^{3}$Departamento de F\'isica, Universidad Cat\'olica del Norte, Avda. Angamos 0610, Casilla 1280, Antofagasta, Chile}

\begin{abstract}
We study cosmological models using dynamical‑systems and averaging methods, encompassing flat and open FLRW geometries as well as the LRS Bianchi types I, III, and V. Under mild regularity and frequency‑scaling assumptions, we obtain a near‑identity conjugacy between the oscillatory flow and an averaged slow flow, with
$\| \mathbf{x}(t)-\bar {\mathbf{x}}(t)\| =\mathcal{O}(H(t))$.
The effective systems preserve the original asymptotics and yield geometry‑dependent late‑time attractor classifications. A corollary addresses the case in which the leading averaged vector field vanishes, so the system exhibits no autonomous drift at order $H^0$.

\end{abstract}

\maketitle

\section{Introduction}
In General Relativity and Cosmology, the Cosmological Principle is satisfied by the homogeneous and isotropic Friedmann--Lemaître--Robertson--Walker (FLRW) metrics.  Additionally, it is believed that cold dark matter, a type of matter that is not visible and only interacts through gravity, 
exists  \citep{Primack:1983xj, Peebles:1984zz, Bond:1984fp, Trimble:1987ee, 
Turner:1991id}. Dark energy is also introduced to explain the Universe's accelerated expansion \citep{Carroll:1991mt}, as supported by numerous observations \citep{SupernovaCosmologyProject:1997zqe, SupernovaSearchTeam:1998fmf, Planck:2018vyg}. The simplest model is $\Lambda$CDM; a paradigm that has proven to be the most successful model describing cosmological observations with fewer constants \citep{Planck:2018vyg}. 

However, despite numerous efforts, no dark matter particles have been directly detected. Their presence is inferred solely from gravitational effects on galaxies and large‑scale structure. Key aspects of their microphysical nature—and that of dark energy—also remain unresolved.

Nevertheless, there is a direction in the literature that deviates from this line of thought and supports the idea that observations can be explained by altering Einstein's equations, leading to modified theories of gravity \citep{Clifton:2011jh, Nojiri:2010wj, DeFelice:2010aj, Capozziello:2011et, CANTATA:2021ktz}. Unified descriptions of the early and late-time Universe included scalar fields \citep{Urena-Lopez:2011gxx, Capozziello:2005tf, Leon:2022oyy}. Moreover, scalar fields are prominent in the physical description of the Universe in inflationary scenario \citep{Guth:1980zm} and can be used as an explanation of late time acceleration of the Universe as a quintessence field used in a generalization of the $\Lambda$CDM model \citep{Ratra:1987rm, Copeland:1999cs}, a phantom scalar field \citep{Nojiri:2005pu, Capozziello:2005tf}, a quintom scalar field model \citep{Cai:2009zp, Guo:2004fq}, a chiral cosmology \citep{Dimakis:2020tzc, Chervon:2013btx}, or multi-scalar field models describing various cosmological epochs \citep{Copeland:1999cs, Tsujikawa:2000wc, Achucarro:2010jv, Akrami:2020zfz}. 
Scalar field models can be studied using qualitative techniques from dynamical systems theory \citep{perko, WE, Coley}, which enable the stability analysis of solutions. A notable example is chaotic inflation, a model of cosmic inflation in which the potential takes the quadratic form \citep{Linde:1983gd, Linde:1986fd}. In addition, asymptotic methods and averaging theory \citep{SandersEtAl2010, Fajman:2020yjb} are valuable for extracting information about the solution space of scalar-field cosmologies, both in vacuum and in the presence of matter \citep{Leon:2019iwj,Leon:2020pfy, Leon:2021lct}. 

A key idea is to construct a time-averaged version of the original system, smoothing out its oscillations \citep{Fajman:2020yjb} while preserving the same late-time behavior. In many of these scenarios, the cosmological dynamics is governed by the Einstein–Klein–Gordon (EKG) system, where the scalar field acts as a dynamical source for the spacetime geometry \citep{WE, Coley}. This framework offers a systematic way to address key conceptual and phenomenological issues in $\Lambda$CDM, including the nature of dark components and the emergence of effective matter sources driving cosmic expansion.

A rich class of scalar-field cosmologies is characterized by the presence of nonlinear or periodic potentials, which induce oscillatory behavior in the scalar field. Such oscillations arise naturally in massive scalar-field models, axion-like scenarios, and post-inflationary dynamics, and they often persist over cosmological time scales \citep{Linde:1983gd, Linde:2002ws,fenichel1979, SandersEtAl2010}. Although the underlying field equations are well defined, the resulting dynamics are intrinsically multiscale: fast oscillations of the scalar field coexist with a slowly evolving cosmological background. From a physical perspective, these microscopic oscillations are typically unobservable, whereas their averaged contribution to the energy–momentum tensor governs the large-scale spacetime dynamics. This separation of time scales and observables motivates the use of averaging methods from dynamical systems theory, which allow for the construction of effective, averaged formulations of the EKG system that describe the slow cosmological evolution without explicitly resolving the fast oscillatory dynamics.
The relevance of this description becomes particularly clear in the late-time regime. In scalar-field cosmologies with oscillatory dynamics, regular oscillations around a stable configuration typically emerge only after initial transient phases—such as kinetic- or curvature-dominated evolution—have decayed. Once this regime is reached, the subsequent evolution is governed by asymptotic solutions of the dynamical system, which fix averaged quantities such as the effective equation of state, the dilution rate of the scalar-field energy density, and the long-term behavior of anisotropies or spatial curvature. From a dynamical-systems perspective, this late-time evolution is often controlled by attractors, invariant manifolds, or scaling solutions that are largely insensitive to initial conditions. These asymptotic states determine the effective matter content, expansion laws, and stability properties of cosmological models \citep{Copeland:1997et, Copeland:2006wr}.

This review is organised as follows. In \S\ref{Sect:2} we give the Einstein–Klein–Gordon field equations for scalar‑field cosmologies with barotropic matter in spatially homogeneous spacetimes, treating (i) LRS Bianchi III, Bianchi I and Kantowski–Sachs, (ii) LRS Bianchi V, and (iii) FLRW (closed, flat, open) reductions. In \S\ref{Sect:3} we review averaging methods and state the Unified Averaging Theorem for LRS Bianchi III, open FLRW, LRS Bianchi I/flat FLRW and LRS Bianchi V. Section \S\ref{sect:4} summarises the main results: attractor classification and bifurcation thresholds $\gamma=\tfrac{2}{3},1$, with a discussion of exceptional cases (Kantowski–Sachs and closed FLRW, the $D$‑normalization).

\section{The field equations}
\label{Sect:2}
We consider scalar-field cosmologies with matter in spatially homogeneous spacetimes, including FLRW models (flat, open, closed), anisotropic Bianchi types I, III, V, and the positively curved Kantowski--Sachs (KS) geometry. Each configuration introduces distinct curvature and anisotropy effects that shape the dynamical evolution and modulate scalar–matter coupling.

\subsection{LRS Bianchi III, Bianchi I, and Kantowski--Sachs}

The metric for these geometries is unified as
\begin{equation}
ds^2 = -dt^2 + e_1^2(t)\,dr^2 + e_2^2(t)\left(d\vartheta^2 + \frac{\sin^2(\sqrt{k}\vartheta)}{k}\,d\zeta^2\right),
\end{equation}
with $k = -1, 0, +1$ corresponding to Bianchi III, Bianchi I, and KS respectively. The Hubble and shear scalars are
\begin{align}
H &= -\tfrac{1}{3}\frac{d}{dt}\ln\left(e_1 e_2^2\right), &
\sigma &= \tfrac{1}{3}\frac{d}{dt}\ln\left(\frac{e_1}{e_2}\right),
\end{align}
and the curvature scalar evolves as
\begin{equation}
K = e_2^{-2}, \qquad {}^{(3)}\!R = 2kK, \qquad \dot{K} = -2(H + \sigma)K.
\end{equation}
The evolution equations of the scalar field and matter are:
\begin{subequations}\label{eq:interacting}
\begin{align}
\ddot{\phi} + 3H\dot{\phi} + V'(\phi) &=0, \\
\dot{\rho}_m + 3\gamma H\rho_m &=0, 
\end{align}
\end{subequations}
The rest of the equations are:
\begin{subequations}\label{eq:anisotropic_system}
\begin{align}
3H^2 + kK &= 3\sigma^2 + \rho_m + \tfrac{1}{2}\dot{\phi}^2 + V(\phi), \\
\dot{\sigma} &= -3H\sigma - \tfrac{kK}{3}, \\
\dot{H} &= -H^2 - 2\sigma^2 - \tfrac{1}{6}(3\gamma - 2)\rho_m - \tfrac{1}{3}\dot{\phi}^2 + \tfrac{1}{3}V(\phi), \\
\dot{e}_1 &= -(H - 2\sigma)e_1.
\end{align}
\end{subequations}

\subsection{LRS Bianchi V}

The Bianchi V metric is given by~\cite{Millano:2023vny}
\begin{equation}
ds^2 = -dt^2 + a^2(t)\,dx^2 + b^2(t)\,e^{2x} \left( dy^2 + \frac{a^4(t)}{b^4(t)} dz^2 \right),
\end{equation}
with Hubble and shear scalars defined as
\begin{equation}
H = \frac{\dot{a}}{a}, \qquad \sigma = \frac{\dot{a}}{a} - \frac{\dot{b}}{b}.
\end{equation}
The evolution equations of scalar field and matter are \eqref{eq:interacting}. The rest of the equations are 
\begin{subequations}\label{eq:BV_system}
\begin{align}
3H^2 & = \sigma^2 + \rho_m + V(\phi) + \frac{3}{a^2}, \\
\dot{\sigma} &= -3H\sigma , \\
\dot{H} &= -H^2 - 2\sigma^2 - \tfrac{1}{6}(3\gamma - 2)\rho_m - \tfrac{1}{3}\dot{\phi}^2 + \tfrac{1}{3}V(\phi), \\
\dot{a} &= aH, \\
\dot{b} &= b(H - \sigma).
\end{align}
\end{subequations}

\subsection{FLRW Geometries}

For isotropic FLRW models with curvature index $k = +1, 0, -1$, the metric is
\begin{equation}
ds^2 = -dt^2 + a^2(t)\left[\frac{dr^2}{1 - kr^2} + r^2(d\vartheta^2 + \sin^2\vartheta\, d\zeta^2)\right].
\end{equation}
The evolution equations of the scalar field and matter are \eqref{eq:interacting}. The rest of the equations are 
\begin{subequations}\label{eq:flrw}
\begin{align}
\dot{a} &= aH, \\
\dot{H} &= -\tfrac{1}{2}(\gamma\rho_m + \dot{\phi}^2) + \tfrac{k}{a^2}, \\
3H^2 &= \rho_m + \tfrac{1}{2}\dot{\phi}^2 + V(\phi) - \tfrac{3k}{a^2}.
\end{align}
\end{subequations}

\section{Averaging in Anisotropic Cosmology}
\label{Sect:3}

In scalar field cosmology, averaging suppresses nonlinear oscillations in the Klein–Gordon system and reveals asymptotic behavior. We apply this to models with generalized harmonic potentials
\begin{equation}
\label{gen-harmonic}
V(\phi) = \mu^2 \phi^2 + f^2(\omega^2 - 2\mu^2)\left(1 - \cos\left(\frac{\phi}{f}\right)\right), \quad \omega^2 > 2\mu^2,
\end{equation}
showing convergence between full and averaged dynamics in LRS Bianchi I, III, V, and open FLRW geometries.

The potential satisfies:
\begin{enumerate}
  \item $V \in C^\infty(\mathbb{R})$, with $\lim_{\phi \to \pm\infty} V(\phi) = +\infty$,
  \item Even symmetry: $V(\phi) = V(-\phi)$,
  \item Global minimum at $\phi = 0$: $V(0) = 0$, $V'(0) = 0$, $V''(0) = \omega^2 > 0$,
  \item Finite set of critical points $\phi_c \neq 0$ satisfying
  \begin{equation}
  2\mu^2 \phi_c + f(\omega^2 - 2\mu^2)\sin\left(\frac{\phi_c}{f}\right) = 0,
  \end{equation}
  \item Bounded extrema: $V_{\min} = 0$, $V_{\max} = \max_{\phi \in [-\phi_*, \phi_*]} V(\phi)$, with $\phi_* = \frac{f(\omega^2 - 2\mu^2)}{2\mu^2}$.
  \item Near $\phi = 0$: $V(\phi) \sim \tfrac{1}{2}\omega^2 \phi^2 + \mathcal{O}(\phi^3)$,
  \item As $|\phi| \to \infty$: $V(\phi) \sim \mu^2 \phi^2 + \mathcal{O}(1)$.
\end{enumerate}

Oscillatory potentials of the form $V(\phi)=V_0\cos(\phi/f)$ arise naturally when the scalar is a pseudo–Nambu–Goldstone boson: a broken global shift symmetry yields a periodic potential of period $2\pi f$. Near a minimum the potential is approximately harmonic, $V\simeq\tfrac{1}{2}m^2\psi^2$ with $\psi=\phi-\phi_{\min}$ and oscillation frequency $\omega\simeq m$; anharmonic corrections are controlled by the small parameter $\psi/f$. This clear separation between fast field oscillations and the slow background evolution driven by $H(t)$ makes such models well suited for averaging.

A commonly used axion potential is
\begin{equation}
V(\phi)=\mu^4\Big[1-\cos\!\big(\tfrac{\phi}{f}\big)\Big],
\end{equation}
which (with appropriate choices of $\mu,f$) models axion dark matter \cite{DAmico:2016jbm}. Variants appear in extensions such as the axionic Einstein–aether setup, e.g.
\begin{equation}
V(\phi,\Phi_*)=\frac{m_A^2\Phi_*^2}{2\pi^2}\Big[1-\cos\!\big(\tfrac{2\pi\phi}{\Phi_*}\big)\Big],
\end{equation}
whose minima occur at $\phi=n\Phi_*$ and which reduces to a quadratic form near each minimum, $V\approx\tfrac{1}{2}m_A^2\psi^2$ for small $\psi$ \cite{Balakin:2020coe}.

Multi‑field axion constructions introduce coupled periodic terms and richer resonance structure. For example,
\begin{equation}
\begin{aligned}
V(\phi_1,\phi_2)&=\mu_1^4\Big[1-\cos\!\big(\tfrac{\phi_1}{f_1}\big)\Big]
+\mu_2^4\Big[1-\cos\!\big(\tfrac{\phi_2}{f_2}\big)\Big]+\mu_3^4\Big[1-\cos\!\big(\tfrac{\phi_1}{f_1}-n\tfrac{\phi_2}{f_2}\big)\Big],
\end{aligned}
\end{equation}
studied with dynamical‑systems and averaging methods to derive effective Einstein–Klein–Gordon dynamics in the $H\to0$ regime \cite{Chakraborty:2021vcr}.

These examples illustrate the typical structure exploited by averaging: locally quadratic minima that support rapid, nearly harmonic oscillations, and periodic or weakly nonlinear corrections that can be averaged out to obtain an effective slow evolution for the cosmological background.

\subsection{Unified Averaging Framework for Scalar-Field Cosmologies}
\label{sect:3}

We introduce the normalized variables
\begin{align}\label{normalized-vars}
\Omega = \sqrt{\frac{\omega^2\phi^2+\dot\phi^2}{6H^2}}, \qquad
\Sigma = \frac{\sigma}{H}, \qquad
\Omega_k = \frac{k}{3H^2}, \qquad
\varphi = \omega t - \arctan\!\left(\frac{\omega\phi}{\dot\phi}\right).
\end{align}

Assume a slowly decaying Hubble parameter $H(t)>0$ with $\lim_{t\to\infty}H(t)=0$.  
In the asymptotic regime $H\to0$, the evolution of the normalized state vector
$\mathbf{x} = (\Omega,\Sigma,\Omega_k,\varphi)^T$ is governed by the expansion
\begin{equation}\label{eq:quasi_standard}
\dot{\mathbf{x}} = \mathbf{f}^0(\mathbf{x},\theta)
+ H(t)\,\mathbf{f}^1(\mathbf{x},\theta)
+ R(\mathbf{x},\theta,H(t)), 
\qquad
\dot{H} = f^{[2]}(\mathbf{x},\theta)\,H^2 + S(\mathbf{x},\theta,H(t)),
\end{equation}
where each $\mathbf{f}^i(\mathbf{x},\theta)$ is $C^1$ in $\mathbf{x}$ and $2\pi$-periodic in
$\theta=\omega t$, and the remainders $R,S$ are higher order in $H$.

The corresponding time averages are
\begin{equation}\label{timeavrg}
\bar{\mathbf{f}}^i(\mathbf{x})
:= \frac{1}{2\pi}\int_0^{2\pi}\mathbf{f}^i(\mathbf{x},\theta)\,d\theta,
\quad i=0,1,
\qquad
\bar{f}^{[2]}(\mathbf{x})
:= \frac{1}{2\pi}\int_0^{2\pi} f^{[2]}(\mathbf{x},\theta)\,d\theta.
\end{equation}
\begin{thm}[Averaging for Scalar-Field Cosmologies]
\label{thm:averaging_scalar_cosmology}
Let $H:[t_x,\infty) \to (0,\infty)$ be $C^1$, strictly decreasing, and satisfy $\lim_{t\to\infty} H(t) = 0$. Consider the system~\eqref{eq:quasi_standard} with $\theta = \omega t$ and $\mathbf{x}$ as in~\eqref{normalized-vars}. Assume: 
\begin{enumerate}[label=(A\arabic*), ref=A\arabic*]
\item \label{A1} \textbf{Smoothness and periodicity.}  
$\mathbf{f}^0$, $\mathbf{f}^1$, and $f^{[2]}$ are $C^1$ in $\mathbf{x}$ on an open set $U \subset \mathbb{R}^4$, and $2\pi$-periodic in $\theta$.

\item \label{A2} \textbf{Controlled remainders.}  
$R = \mathcal{O}(H^2)$ and $S = \mathcal{O}(H^3)$ uniformly on compact subsets of $U \times \mathbb{S}^1$.

\item \label{A3} \textbf{Well-defined averages.}  
The time averages defined in~\eqref{timeavrg} exist and are Lipschitz continuous on $U$.

\item \label{A4} \textbf{Matched initial data.}  
The full and averaged systems share initial data at $t = t_x$, with $H(t_x) = \varepsilon \ll 1$.

\item \label{A5} \textbf{Asymptotic stability.}  
The averaged system $\dot{\bar{\mathbf{x}}} = \bar{\mathbf{f}}^0(\bar{\mathbf{x}})$ admits an asymptotically stable equilibrium $\mathbf{x}_* \in U$.

\item \label{A6} \textbf{Frequency–amplitude scaling.}  
The frequency satisfies $\omega^{-1} = \mathcal{O}(\varepsilon) = \mathcal{O}(H(t_x))$, ensuring that boundary terms from integration by parts are absorbed into the remainders.
\end{enumerate}

Then the solutions $\mathbf{x}(t)$ and $\bar{\mathbf{x}}(t)$ with common initial data satisfy
\begin{equation}
\|\mathbf{x}(t) - \bar{\mathbf{x}}(t)\| = \mathcal{O}(H(t)) \qquad \text{as } t \to \infty,
\label{error}
\end{equation}
and both converge to $\mathbf{x}_*$.
\end{thm}
\begin{proof}
Decompose $\mathbf{f}^i = \bar{\mathbf{f}}^i + \tilde{\mathbf{f}}^i$ for $i = 0,1$, where $\tilde{\mathbf{f}}^i$ has zero average in $\theta$ as in~\eqref{timeavrg}. Define $\bar{R}(\mathbf{x}, H) := \tfrac{1}{2\pi} \int_0^{2\pi} R(\mathbf{x}, \theta, H)\,d\theta = \mathcal{O}(H^2)$ by~\ref{A2}.
Let $\bar{\mathbf{x}}(t)$ solve the averaged system
$\dot{\bar{\mathbf{x}}} = \bar{\mathbf{f}}^0(\bar{\mathbf{x}}) + H(t)\,\bar{\mathbf{f}}^1(\bar{\mathbf{x}}) + \bar{R}(\bar{\mathbf{x}}, H(t)),$
with initial data $\bar{\mathbf{x}}(t_x) = \mathbf{x}(t_x)$ by~\ref{A4}, and  error $\mathbf{e}(t) := \mathbf{x}(t) - \bar{\mathbf{x}}(t)$. 

Subtracting the systems yields
\begin{equation}
\dot{\mathbf{e}} = A(t)\,\mathbf{e} + \tilde{\mathbf{f}}^0(\mathbf{x}, \omega t) + H(t)\,\tilde{\mathbf{f}}^1(\mathbf{x}, \omega t) + \Delta_R(t),
\end{equation}
where $A(t)$ is the Jacobian of the averaged vector field, with $\|A(t)\| \le L$ by~\ref{A3}, and $\|\Delta_R(t)\| = \mathcal{O}(H^2) + \mathcal{O}(\|\mathbf{e}\| H)$.
By variation of constants,
\begin{equation}
\mathbf{e}(t) = \int_{t_x}^t \Phi(t,s)\left[\tilde{\mathbf{f}}^0(\mathbf{x}(s), \omega s) + H(s)\,\tilde{\mathbf{f}}^1(\mathbf{x}(s), \omega s) + \Delta_R(s)\right]\,ds,
\end{equation}
where $\Phi(t,s)$ is the evolution operator of the linearized system, satisfying $\|\Phi(t,s)\| \le e^{L(t-s)}$, and $\mathbf{e}(t_x) = 0$. The oscillatory integrals are estimated by integration by parts, using the $C^1$ regularity in $\theta$ from~\ref{A1} and boundedness of trajectories. This yields convolution terms bounded by $C_1 \int_{t_x}^t e^{L(t-s)} H(s)\,ds$, plus boundary terms of order $H/\omega$, which are $\mathcal{O}(H^2)$ by~\ref{A6} and absorbed into $\Delta_R$.
Thus, for some constant $C > 0$,
\begin{equation}
\|\mathbf{e}(t)\| \le C \left( \int_{t_x}^t e^{L(t-s)} H(s)\,ds + \int_{t_x}^t e^{L(t-s)} \|\mathbf{e}(s)\| H(s)\,ds \right).
\end{equation}

Choose $t_0 \ge t_x$ such that $C H(t_0) \le \tfrac{1}{2}L$, which is possible since $H(t) \to 0$. For $t \ge t_0$, monotonicity of $H$ implies
$\int_{t_0}^t e^{L(t-s)} H(s)\,ds \le H(t) \int_{t_0}^t e^{L(t-s)}\,ds \le \tfrac{1}{L} H(t),$
so
\begin{equation}
\sup_{s \in [t_0, t]} \|\mathbf{e}(s)\| \le \tfrac{C}{L} H(t) + \tfrac{C}{L} H(t) \sup_{s \in [t_0, t]} \|\mathbf{e}(s)\|.
\end{equation}
Choosing $t_0$ so that $\tfrac{C}{L} H(t_0) \le \tfrac{1}{2}$, we obtain
\begin{equation}
\sup_{s \in [t_0, t]} \|\mathbf{e}(s)\| \le \tfrac{2C}{L} H(t),
\end{equation}
\noindent
hence $\|\mathbf{e}(t)\| = \mathcal{O}(H(t))$ as $t \to \infty$. On the compact interval $[t_x, t_0]$, the error remains uniformly bounded by continuity and the smallness of $\varepsilon = H(t_x)$.

Finally, by~\ref{A5}, the averaged trajectory $\bar{\mathbf{x}}(t)$ converges to $\mathbf{x}_*$, and since $\|\mathbf{x}(t) - \bar{\mathbf{x}}(t)\| = \mathcal{O}(H(t))$, it follows that $\mathbf{x}(t) \to \mathbf{x}_*$ as well.
\end{proof}
\begin{corollary}[Degenerate Averaging Regime]
\label{corollary}
Assume \ref{A1}–\ref{A6} and suppose $\mathbf{f}^0 \equiv 0$. If the model parameters determine $\omega$ such that $\omega^{-1} = \mathcal{O}(\varepsilon)$, then for common initial data at $t = t_x$, the solutions $\mathbf{x}(t)$ and $\bar{\mathbf{x}}(t)$ satisfy the estimate~\eqref{error} and both converge to the attracting equilibrium $\mathbf{x}_*$ of the frozen slow field $\dot{\bar{\mathbf{x}}} = \bar{\mathbf{f}}^1(\bar{\mathbf{x}})$.
\end{corollary}
\begin{proof}
When $\mathbf{f}^0 \equiv 0$, the leading-order dynamics are governed by the averaged field $H \bar{\mathbf{f}}^1$. The same argument as in Theorem~\ref{thm:averaging_scalar_cosmology} applies, with the roles of $\bar{\mathbf{f}}^0$ and $\bar{\mathbf{f}}^1$ interchanged. The scaling condition $\omega^{-1} = \mathcal{O}(H(t_x))$ ensures that boundary terms from integration by parts remain $\mathcal{O}(H^2)$ and are absorbed into the remainders. The result follows.
\end{proof}
\subsubsection{Application to LRS Bianchi III Cosmology}

Let $\mathbf{x} = (\Omega, \Sigma, \Omega_k, \Phi)^T$ be the state vector for a scalar-field cosmology with matter. Imposing the modulation-frequency relation
$b \mu^3 + 2 f \mu^2 - f \omega^2 = 0$, eliminates the zeroth-order term $\mathbf{f}^0$, yielding the reduced system:
\begin{align}
\dot{\mathbf{x}} &= H {\footnotesize \left(\begin{array}{c}
\frac{1}{2} \Omega \Big(-3 (\gamma - 2) \Sigma^2 + (2 - 3\gamma) \Omega_k + 3(\Omega^2 - 1)(-\gamma + 2 \cos^2(t \omega - \Phi))\Big) \\
\frac{1}{2} \Bigg( \Omega_k ((2 - 3\gamma) \Sigma + 2) + 3 \Sigma \Big(-(\gamma - 2) \Sigma^2 + \gamma + \Omega^2 (-\gamma + 2 \cos^2(t \omega - \Phi)) - 2\Big)\Bigg) \\
\Omega_k \Big(-3\gamma(\Sigma^2 + \Omega^2 + \Omega_k - 1) + 6 \Sigma^2 - 2 \Sigma + 6 \Omega^2 \cos^2(t \omega - \Phi) + 2 \Omega_k - 2\Big) \\
-\frac{3}{2} \sin(2 t \omega - 2 \Phi)
\end{array}\right)} + \mathcal{O}(H^2), \label{equx2} \\
\dot{H} &= -\frac{3}{2} H^2 \Big(\gamma(1 - \Sigma^2 - \Omega_k - \Omega^2) + 2 \Sigma^2 + \frac{2}{3} \Omega_k + 2 \Omega^2 \cos^2(t \omega - \Phi) \Big) + \mathcal{O}(H^3), \label{EQ:81b}
\end{align}
To facilitate asymptotic analysis, we introduce the logarithmic time variable $\tau$, defined by $d\tau = H\,dt$. This transformation normalizes the system with respect to the expansion rate and isolates the slow evolution of the averaged variables. The averaged dynamics in Bianchi~III geometry then become:
\begin{subequations}
\label{eq:tau_bianchiIII}
\begin{align}
\frac{d\overline{\Omega}}{d\tau} &= \frac{1}{2} \overline{\Omega} \left[ -3\gamma(\overline{\Sigma}^2 + \overline{\Omega}^2 + \overline{\Omega}_k - 1) + 6\overline{\Sigma}^2 + 3\overline{\Omega}^2 + 2\overline{\Omega}_k - 3 \right], \label{tauIIIeq:Omega} \\
\frac{d\overline{\Sigma}}{d\tau} &= \frac{1}{2} \left\{ \overline{\Sigma} \left[ -3\gamma(\overline{\Sigma}^2 + \overline{\Omega}^2 + \overline{\Omega}_k - 1) + 6\overline{\Sigma}^2 + 3\overline{\Omega}^2 + 2\overline{\Omega}_k - 6 \right] + 2\overline{\Omega}_k \right\}, \label{tauIIIeq:Sigma} \\
\frac{d\overline{\Omega}_k}{d\tau} &= -\overline{\Omega}_k \left[ 3\gamma(\overline{\Sigma}^2 + \overline{\Omega}^2 + \overline{\Omega}_k - 1) - 6\overline{\Sigma}^2 + 2\overline{\Sigma} - 3\overline{\Omega}^2 - 2\overline{\Omega}_k + 2 \right], \label{tauIIIeq:Omegak}
\end{align}
\end{subequations}
In all geometric models considered, the evolution of the averaged angle remains trivial:
\begin{equation}
\frac{d\overline{\Phi}}{d\tau} = 0. \label{tauIIIeq:Phi}
\end{equation}
In this model we can set $\mathbf f^0(\mathbf x)\equiv 0$ with modulation-frequency relation $b\mu^3 + 2f\mu^2 - f\omega^2 = 0$. 
Then, applying Corollary \ref{corollary}, the full and averaged solutions satisfy
\eqref{error} as  $t \to \infty$. Smooth Transformation Near $H = 0$ are found in {\cite[Thm.~2]{Leon:2021lct}}: 
\begin{thm}[Smooth Transformation Near $H = 0$, {\cite[Thm.~2]{Leon:2021lct}}]
There exists a smooth transformation
\begin{equation}
\label{AppBIIIquasilinear211}
\mathbf{x} = \psi(\mathbf{x}_0) := \mathbf{x}_0 + H \mathbf{g}(H, \mathbf{x}_0, t),
\end{equation}
with $\mathbf{x}_0 = (\Omega_0, \Sigma_0, \Omega_{k0}, \Phi_0)^T$ and $\mathbf{g} \in C^1$, such that both full and averaged solutions converge to the same asymptotic state. The invariant set $\Sigma = 0$ recovers the negatively curved FLRW model.
\end{thm}
We obtained $\mathbf{g}$ by construction, and it is smooth, bounded, and periodic. A linear stability analysis, the Centre Manifold calculations, and combined with previous results, lead to: 
\begin{thm}[Late-Time Attractors, {\cite[Thm.~3]{Leon:2021lct}}]
For barotropic index $ \gamma \in (0,2] $, the attractors are: (i) matter-dominated FLRW universe $ \mathcal{F}_0 $, for $ 0 < \gamma \leq 2/3 $; (ii) matter-curvature scaling solution $ \mathcal{MC} $, for $ 2/3 < \gamma < 1 $; (iii) Bianchi~III flat spacetime $ \mathcal{D} $, for $ 1 \leq \gamma \leq 2 $.
\end{thm}

\subsubsection{Open FLRW ($ k = -1 $)}
The cosmological equations for open FLRW ($ k = -1 $) can be written as 
\begin{align}
& \dot{\mathbf{x}}= H \mathbf{f}(t, \mathbf{x}) + \mathcal{O}(H^2), \;   \mathbf{x}= \left(\Omega, \Omega_k, \Phi \right)^T,\\
   &  \dot{H}=-H^2\Bigg[ \frac{1}{2} \left(3
   \gamma\left(1-\Omega^2- \Omega_{k}\right)+2 \Omega_{k}\right)  +3 \Omega^2 \cos ^2(t \omega -\Phi)\Bigg]+  \mathcal{O}(H^3),
\\
    \label{EQ:50}
 &  f(t, \mathbf{x}) = 
   {\footnotesize \left(\begin{array}{c}
   \frac{1}{2}   \Omega  \left(3 \gamma -3 \gamma  \left(\Omega^2+\Omega_{k}\right)+2 \Omega_{k}\right)  +3  \Omega \left(\Omega^2-1\right) \cos ^2(t \omega -\Phi) \\
- \Omega_{k} \left(3 \gamma  \Omega^2+(3 \gamma -2) (\Omega_{k}-1)\right)   +6  \Omega^2 \Omega_{k} \cos ^2(t \omega -\Phi)\\
    -\frac{3}{2} \sin (2 t \omega -2 \Phi)
      \end{array}
   \right)}.
\end{align}
Replacing $\dot{\mathbf{x}}= H \mathbf{f}(t, \mathbf{x})$ with  $\mathbf{f}(t, \mathbf{x})$ defined by \eqref{EQ:50} with $\dot{\mathbf{y}}= H  \overline{f}(\mathbf{y})$, $\mathbf{y}= \left(\overline{\Omega}, {\overline{\Omega}_{k}}, \overline{\Phi} \right)^T$  and using the time averaging \eqref{timeavrg} we obtain the time-averaged system: 
\begin{subequations}
\label{eq:tau_system}
\begin{align}
\frac{d\overline{\Omega}}{d\tau} &= -\frac{1}{2} \overline{\Omega} \left[ 3(\gamma - 1)(\overline{\Omega}^2 - 1) + (3\gamma - 2)\, \overline{\Omega}_k \right], \label{eq:tau_Omega} \\
\frac{d\overline{\Omega}_k}{d\tau} &= -\overline{\Omega}_k \left[ 3(\gamma - 1)\overline{\Omega}^2 - 3\gamma + (3\gamma - 2)\, \overline{\Omega}_k + 2 \right].\label{eq:tau_Omegak}
\end{align}
\end{subequations}
The evolution of the averaged angle remains trivial \eqref{tauIIIeq:Phi}.
Under assumptions~\hyperref[A1]{(A1)}–\hyperref[A3]{(A3)}, the full and averaged solutions satisfy
\eqref{error} as  $t \to \infty$, with modulation-frequency relation $b\mu^3 + 2f\mu^2 - f\omega^2 = 0$ (Thm.~ \ref{thm:averaging_scalar_cosmology}). Moreover, late-time attractors, {\cite[Thm.~4]{Leon:2021lct}} are identified: 
\begin{thm}[Late-Time Attractors, {\cite[Thm.~4]{Leon:2021lct}}]
For barotropic index $ \gamma \in (0,2] $, the attractors are: (i) matter-dominated FLRW universe $ \mathcal{F}_0 $, for $ 0 < \gamma \leq 2/3 $; (ii) Milne solution $ \mathcal{C} $, for $ 2/3 < \gamma < 2 $.
\end{thm}

\subsubsection{Asymptotic Dynamics and Averaging Results in LRS Bianchi I}

The full system in LRS Bianchi I geometry admits the expansion:
\begin{align}
\dot{\mathbf{x}} &= H \mathbf{f}(\mathbf{x}, t) + \mathcal{O}(H^2), \label{equx} \\
\dot{H} &= -\frac{3}{2} H^2 \left[ \gamma (1 - \Sigma^2 - \Omega^2) + 2\Sigma^2 + 2\Omega^2 \cos^2(\Phi - \omega t) \right] + \mathcal{O}(H^3), \label{EQ:61b}
\end{align}
where \eqref{EQ:61b} is the Raychaudhuri equation and the oscillatory vector field is given by:
\begin{align}
\label{EQ:87}
\mathbf{f}(\mathbf{x}, t) =
\begin{pmatrix}
\frac{3}{2} \Omega \left[ \gamma (1 - \Sigma^2 - \Omega^2) + 2\Sigma^2 + 2(\Omega^2 - 1) \cos^2(\Phi - \omega t) \right] \\\\
\frac{3}{2} \Sigma \left[ -\gamma (\Sigma^2 + \Omega^2 - 1) + 2\Sigma^2 + 2\Omega^2 \cos^2(\Phi - \omega t) - 2 \right] \\\\
-\frac{3}{2} \sin(2\omega t - 2\Phi)
\end{pmatrix}.
\end{align}

Replacing $\dot{\mathbf{x}} = H \mathbf{f}(\mathbf{x}, t)$ with its time-averaged counterpart $\dot{\mathbf{y}} = H \bar{\mathbf{f}}(\mathbf{y})$, where $\mathbf{y} = (\bar{\Omega}, \bar{\Sigma}, \bar{\Phi})^T$, we obtain the averaged system:
\begin{subequations}
\label{eq:tau_bianchiI}
\begin{align}
\frac{d\bar{\Omega}}{d\tau} &= \frac{3}{2} \bar{\Omega} \left[ \gamma (1 - \bar{\Sigma}^2 - \bar{\Omega}^2) + 2\bar{\Sigma}^2 + \bar{\Omega}^2 - 1 \right], \label{taueq:Omega} \\
\frac{d\bar{\Sigma}}{d\tau} &= \frac{3}{2} \bar{\Sigma} \left[ \gamma (1 - \bar{\Sigma}^2 - \bar{\Omega}^2) + 2\bar{\Sigma}^2 + \bar{\Omega}^2 - 2 \right]. \label{taueq:Sigma}
\end{align}
\end{subequations}
The evolution of the averaged angle remains trivial \eqref{tauIIIeq:Phi}. 
\begin{thm}[Smooth Transformation Near $H = 0$, {\cite[Thm.~ 1]{Leon:2021rcx}}]
There exists a smooth transformation
\begin{equation}
\mathbf{x}(t) = \bar{\mathbf{x}}(t) + H(t)\,\mathbf{g}(H, \bar{\mathbf{x}}, t),
\end{equation}
such that the full and averaged systems share the same asymptotic behavior.
\end{thm}
\begin{thm}[Late-Time Attractors, {\cite[Thm.~2]{Leon:2021rcx}}]
For barotropic index $\gamma \in (0,2]$, the attractors are:
(i) $\mathcal{F}_0$: flat matter-dominated FLRW universe, for $0 < \gamma < 1$,
(ii) $\mathcal{F}$: scalar-field dominated solution, for $1 < \gamma \leq 2$.
\end{thm}

\subsubsection{Asymptotic Dynamics in Flat FLRW Geometry}

For the flat FLRW universe with $k = 0$ and $\gamma \neq 1$ is:
\begin{align}\label{avrgsystFLRW}
\frac{d\bar{\Omega}}{d \tau} &= -\frac{3}{2} \bar{\Omega} (\gamma - 1)(\bar{\Omega}^2 - 1) \implies \displaystyle \bar{\Omega}(\tau) = \frac{ \Omega_0 e^{\frac{3\gamma \tau}{2}} }{ \sqrt{ \Omega_0^2 e^{3\gamma \tau} + e^{3\tau}(1 - \Omega_0^2) } }, \quad \bar{\Omega}(0) = \Omega_0.
\end{align}
The evolution of the averaged angle remains trivial \eqref{tauIIIeq:Phi}.
\begin{thm}[Late-Time Attractors, {\cite[Thm.~2]{Leon:2021rcx}}]
For barotropic index $\gamma \in (0,2]$, the attractors of the averaged flat FLRW system are:
(i) Matter-dominated FLRW universe $\mathcal{F}_0$, corresponding to $\bar{\Omega} = 0$, with eigenvalue $-\frac{3}{2}(1 - \gamma)$; it is a \textit{sink} for $0 < \gamma < 1$, and a \textit{source} for $1 < \gamma \leq 2$.
 (ii) Scalar-field dominated solution $\mathcal{F}$, corresponding to $\bar{\Omega} = 1$, with eigenvalue $3(1 - \gamma)$; it is a \textit{source} for $0 < \gamma < 1$, and a \textit{sink} for $1 < \gamma \leq 2$.
\end{thm}

\subsubsection{Averaged Dynamics in Bianchi Type-V Geometry}
\label{sect:5}

We extend the averaging framework to the anisotropic Bianchi type-V model~\cite{Millano:2023vny} for the scalar potential
\begin{equation}
V(\phi) = \frac{1}{2} \phi^2 + f\left(1 - \cos\left(\frac{\phi}{f}\right)\right), \quad \omega^2 > 1,
\end{equation}
which supports bounded oscillations and facilitates averaging analysis in curved anisotropic geometries.

To analyze the asymptotic behavior in Bianchi~V geometry, we begin with the full system:
\begin{subequations}
\label{eq:full_bianchiV}
\begin{align}
\dot{H} &= -\frac{3}{2} H^2 \left[ \gamma(1 - \Sigma^2 - \Omega_k - \Omega^2) + 2\Sigma^2 + \frac{2}{3}\Omega_k + 2\Omega^2 \cos^2(\omega t - \Phi) \right] + \mathcal{O}(H^3), \label{eq:fullHV} \\
\dot{\mathbf{x}} &= H\,\mathbf{f}(\mathbf{x}, t) + \mathcal{O}(H^2), \quad \mathbf{x} = (\Omega, \Sigma, \Omega_k, \Phi)^T, \label{eq:fullxV}
\end{align}
\end{subequations}
where the oscillatory vector field is given by:
\begin{align}
\label{eq:fullfV}
\mathbf{f}(\mathbf{x}, t) = {\footnotesize
\begin{pmatrix}
\frac{1}{2} \Omega \left[ -3(\gamma - 2)\Sigma^2 + (2 - 3\gamma)\Omega_k + 3(\Omega^2 - 1)(- \gamma + 2\cos^2(\omega t - \Phi)) \right] \\\\
\frac{1}{2} \left[ \Omega_k((2 - 3\gamma)\Sigma + 2) + 3\Sigma \left( -(\gamma - 2)\Sigma^2 + \gamma + \Omega^2(-\gamma + 2\cos^2(\omega t - \Phi)) - 2 \right) \right] \\\\
\Omega_k \left[ -3\gamma(\Sigma^2 + \Omega^2 + \Omega_k - 1) + 6\Sigma^2 - 2\Sigma + 6\Omega^2 \cos^2(\omega t - \Phi) + 2\Omega_k - 2 \right] \\\\
-\frac{3}{2} \sin(2\omega t - 2\Phi)
\end{pmatrix}.}
\end{align}

Replacing $\dot{\mathbf{x}} = H\,\mathbf{f}(\mathbf{x}, t)$ with its time-averaged counterpart $\dot{\mathbf{y}} = H\,\overline{\mathbf{f}}(\mathbf{y})$, where $\mathbf{y} = (\overline{\Omega}, \overline{\Sigma}, \overline{\Omega}_k, \overline{\Phi})^T$, and using the time averaging \eqref{timeavrg} we obtain the time-averaged system:
\begin{subequations}
\label{eq:avg_bianchiV_components}
\begin{align}
\frac{dH}{d\tau} &= -\frac{1}{2} H \left[ 3\gamma(1 - \overline{\Sigma}^2 - \overline{\Omega}^2 - \overline{\Omega}_k) + 6\overline{\Sigma}^2 + 3\overline{\Omega}^2 + 2\overline{\Omega}_k \right], \label{eq:tau_HV} \\
\frac{d\overline{\Omega}}{d\tau} &= \frac{1}{2} \overline{\Omega} \left[ -3\gamma(\overline{\Sigma}^2 + \overline{\Omega}^2 + \overline{\Omega}_k - 1) + 6\overline{\Sigma}^2 + 3\overline{\Omega}^2 + 2\overline{\Omega}_k - 3 \right], \label{eq:tau_OmegaV} \\
\frac{d\overline{\Sigma}}{d\tau} &= \frac{1}{2} \left\{ \overline{\Sigma} \left[ -3\gamma(\overline{\Sigma}^2 + \overline{\Omega}^2 + \overline{\Omega}_k - 1) + 6\overline{\Sigma}^2 + 3\overline{\Omega}^2 + 2\overline{\Omega}_k - 6 \right] + 2\overline{\Omega}_k \right\}, \label{eq:tau_SigmaV} \\
\frac{d\overline{\Omega}_k}{d\tau} &= -\overline{\Omega}_k \left[ 3\gamma(\overline{\Sigma}^2 + \overline{\Omega}^2 + \overline{\Omega}_k - 1) - 6\overline{\Sigma}^2 + 2\overline{\Sigma} - 3\overline{\Omega}^2 - 2\overline{\Omega}_k + 2 \right]. \label{eq:tau_Omegakv} 
\end{align}
\end{subequations}
The evolution of the averaged angle remains trivial \eqref{tauIIIeq:Phi}.

\begin{thm}[Smooth Transformation Near $ H = 0 $, {\cite[Thm.~ 2]{Millano:2023vny}}]
\label{Thm3.9}
There exists a smooth transformation
\begin{equation}
\mathbf{x}(t) = \bar{\mathbf{x}}(t) + H(t)\,\mathbf{g}(H, \bar{\mathbf{x}}, t),
\end{equation}
such that the full and averaged systems share the same asymptotic behavior.
\end{thm}
\begin{thm}[Averaging Validity, {\cite[Thm.~ 1]{Millano:2023vny}}]
Let $ \mathbf{x}(t) $ solve a scalar field cosmology system in LRS Bianchi~V geometry. Then, the full and averaged solutions satisfy
\begin{equation}
\| \mathbf{x}(t) - \bar{\mathbf{x}}(t) \| = \mathcal{O}(H(t)), \quad \text{as } t \to \infty,
\end{equation}
with modulation-frequency relation  $f = (\omega^2 - 1)^{-1}$.
\end{thm} 
Late-Time Attractors, {\cite[Thm.~3]{Millano:2023vny}} are the following
\begin{thm}[Late-Time Attractors, {\cite[Thm.~3]{Millano:2023vny}}]
For barotropic index $ \gamma \in (0,2] $, the attractors are: (i) matter-dominated FLRW universe $ \mathcal{F}_0 $, for $ 0 < \gamma \leq 2/3 $; (ii) matter-curvaturej scaling solution $ \mathcal{MC} $, for $ 2/3 < \gamma < 1 $; (iii) scalar-field dominated solution $ \mathcal{F} $, for $ 1 < \gamma \leq 2 $.
\end{thm}
\begin{table}[H]
\centering
\small
\setlength{\tabcolsep}{4pt}
\renewcommand{\arraystretch}{1.1}
\begin{tabular}{clll}
\toprule
\textbf{Label} & $a(t)$ & $b(t)$ & \textbf{Geometry Type} \\
\midrule
$\mathcal{T}$ & $\frac{3 H_0 t + 1}{c_2}$ & $\sqrt{c_2}$ & Taub–Kasner \\
$\mathcal{Q}$ & $c_1^{-2}(3 H_0 t + 1)^{-1/3}$ & $c_1^{-1}(3 H_0 t + 1)^{2/3}$ & Non-flat LRS Kasner (Bianchi I) \\
$\mathcal{D}$ & $c_1^{-1}$ & $\frac{3 H_0 t + 2}{2 \sqrt{c_1}}$ & Bianchi III flat spacetime \\
$\mathcal{F}$ & $c_1^{-1} t^{2/3}$ & $c_2^{-1/2} t^{2/3}$ & Einstein–de Sitter \\
$\mathcal{F}_0$ & $\ell_0 \left( \frac{3 \gamma H_0 t}{2} + 1 \right)^{\frac{2}{3\gamma}}$ & same & FLRW matter dominated \\
$\mathcal{MC}$ & same & same & Matter–curvature scaling \\
$\mathcal{C}$ & $a_0(H_0 t + 1)$ & — & Milne solution \\
$\mathcal{V}_0$ & $\ell_1 \left( \frac{3 \gamma H_0 t}{2} + 1 \right)^{\frac{2}{3\gamma}}$ & $\ell_2 \left( \frac{3 \gamma H_0 t}{2} + 1 \right)^{\frac{2}{3\gamma}}$ & Bianchi V FLRW-like \\
$\mathcal{V}_{\text{MC}}$ & $a_0 t^{2/3}$ & $b_0 t^{2/3}$ & Bianchi V curvature scaling \\
\bottomrule
\end{tabular}
\caption{Metric forms associated with late-time attractors. Bianchi V forms updated from~\cite{Millano:2023vny}.}
\label{tab:metric_attractors}
\end{table}
The late-time attractors of scalar-field cosmologies with generalized harmonic potentials and barotropic matter depend on the geometry and the barotropic index $\gamma \in (0,2]$ as follows:
\begin{itemize}
  \item \textbf{Flat FLRW}: $\mathcal{F}_0$ for $0 < \gamma < 1$; $\mathcal{F}$ for $1 < \gamma \leq 2$.
  \item \textbf{Open FLRW}: $\mathcal{F}_0$ for $0 < \gamma \leq \tfrac{2}{3}$; $\mathcal{C}$ for $\tfrac{2}{3} < \gamma < 2$.
  \item \textbf{LRS Bianchi I}~\cite{Millano:2025vjo}: $\mathcal{F}_0$ for $0 < \gamma < 1$; $\mathcal{F}$ for $1 < \gamma \leq 2$.
  \item \textbf{LRS Bianchi III}: $\mathcal{F}_0$ for $0 < \gamma \leq \tfrac{2}{3}$; $\mathcal{MC}$ for $\tfrac{2}{3} < \gamma < 1$; $\mathcal{D}$ for $1 \leq \gamma \leq 2$.
  \item \textbf{LRS Bianchi V}~\cite{Millano:2023vny}: $\mathcal{V}_0$ for $0 < \gamma \leq \tfrac{2}{3}$; $\mathcal{V}_{\text{MC}}$ for $\tfrac{2}{3} < \gamma < 1$; $\mathcal{F}$ for $1 < \gamma \leq 2$.
\end{itemize}

In Table~\ref{tab:metric_attractors} we list the metric forms of the late‑time attractors for the homogeneous models studied (Bianchi V entries updated from~\cite{Millano:2023vny}). The table encapsulates the main conclusion of the averaged dynamical analysis: spatial geometry together with the barotropic index $\gamma$ selects the asymptotic balance—matter, curvature/anisotropy, or scalar field—that controls expansion. For small $\gamma$ the matter‑dominated FLRW scaling $\mathcal{F}_0$ with $a(t)\propto t^{2/(3\gamma)}$ is preferred, while larger $\gamma$ favors the scalar‑field (stiff‑like) scaling $\mathcal{F}$ (here shown as $t^{2/3}$). Open models can instead approach the curvature‑dominated Milne state $\mathcal{C}$ when curvature decays more slowly than matter.

The averaged framework clarifies anisotropy: LRS Bianchi families admit FLRW‑like limits ($\mathcal{V}_0$) or curvature/anisotropy–matter scaling states ($\mathcal{MC},\mathcal{V}_{\mathrm{MC}}$), depending on whether shear and curvature decay relative to the fluid. Vacuum or shear‑dominated solutions ($\mathcal{T},\mathcal{Q}$) lie on phase‑space boundaries and typically appear as transient or measure‑zero asymptotics rather than generic attractors. Multiplicative constants $c_i,\ell_i,a_0,b_0,H_0$ set amplitudes and time scales but do not alter the qualitative stability classification.

Two practical consequences follow. First, under the hypotheses of the Unified Averaging Theorem, the averaged system approximates the full oscillatory dynamics with error $O(H(t))$ as $H(t)\to0$, and a smooth near‑identity transformation relates the two flows; hence, the averaged fixed points reliably indicate late‑time scaling laws. Second, the values $\gamma=\tfrac{2}{3}$ and $\gamma=1$ mark the principal bifurcations separating matter domination, curvature/anisotropy–matter scaling, and scalar‑field dominance, and therefore guide the linear stability and centre‑manifold analyses. The averaging results apply to flat and open FLRW and to LRS Bianchi I, III, and V under the stated hypotheses; Kantowski–Sachs and closed FLRW require separate treatment because $H(t)$ need not decay monotonically there.
\begin{figure}[t]
  \centering
  \begin{subfigure}[b]{0.30\textwidth}
    \centering
    \includegraphics[width=\textwidth]{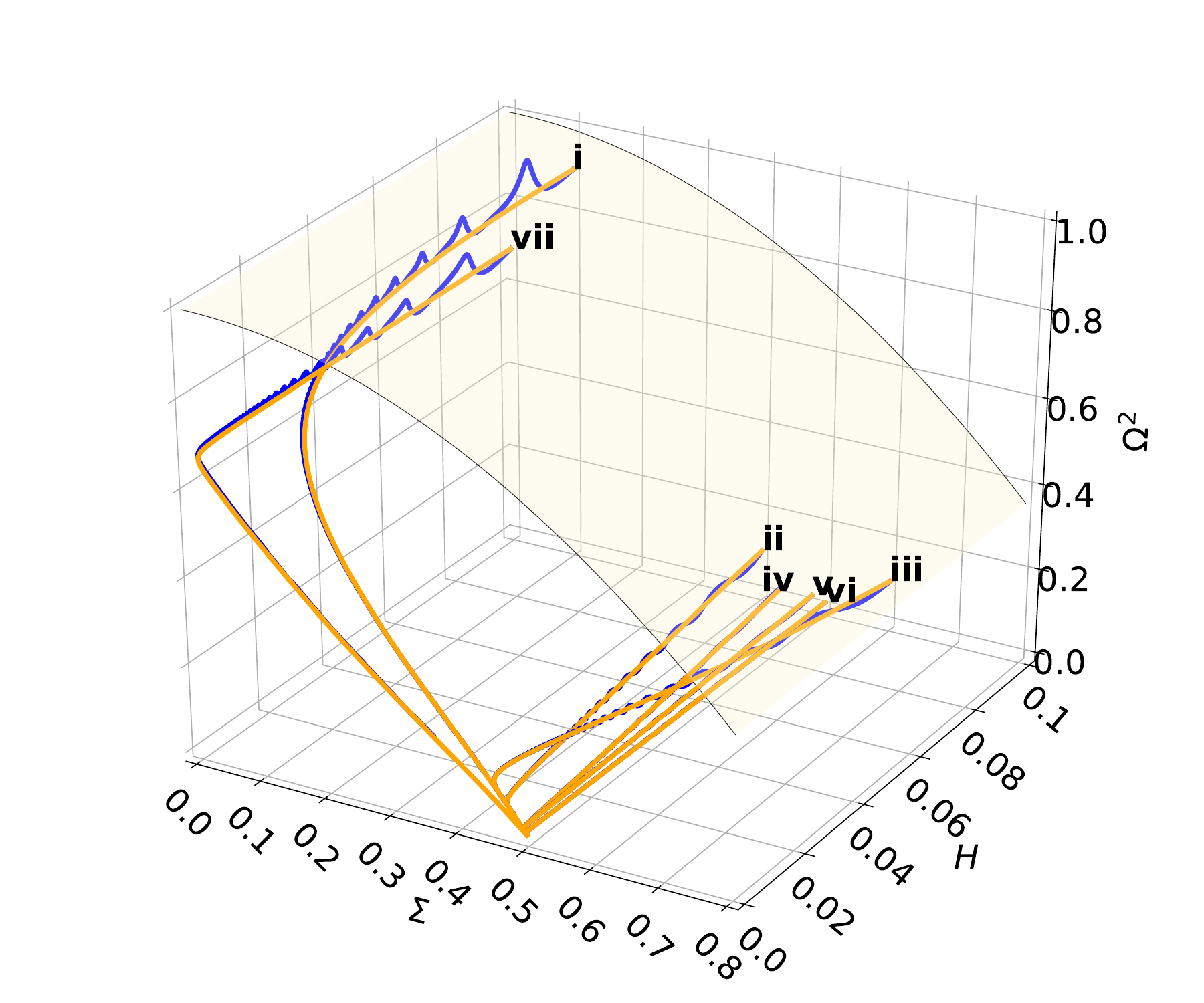}
    \caption{LRS Bianchi III with $\gamma=1$: projection in $(\Sigma,H,\Omega^2)$.}
    \label{fig:BI3_S}
  \end{subfigure}\hfill
  \begin{subfigure}[b]{0.30\textwidth}
    \centering
    \includegraphics[width=\textwidth]{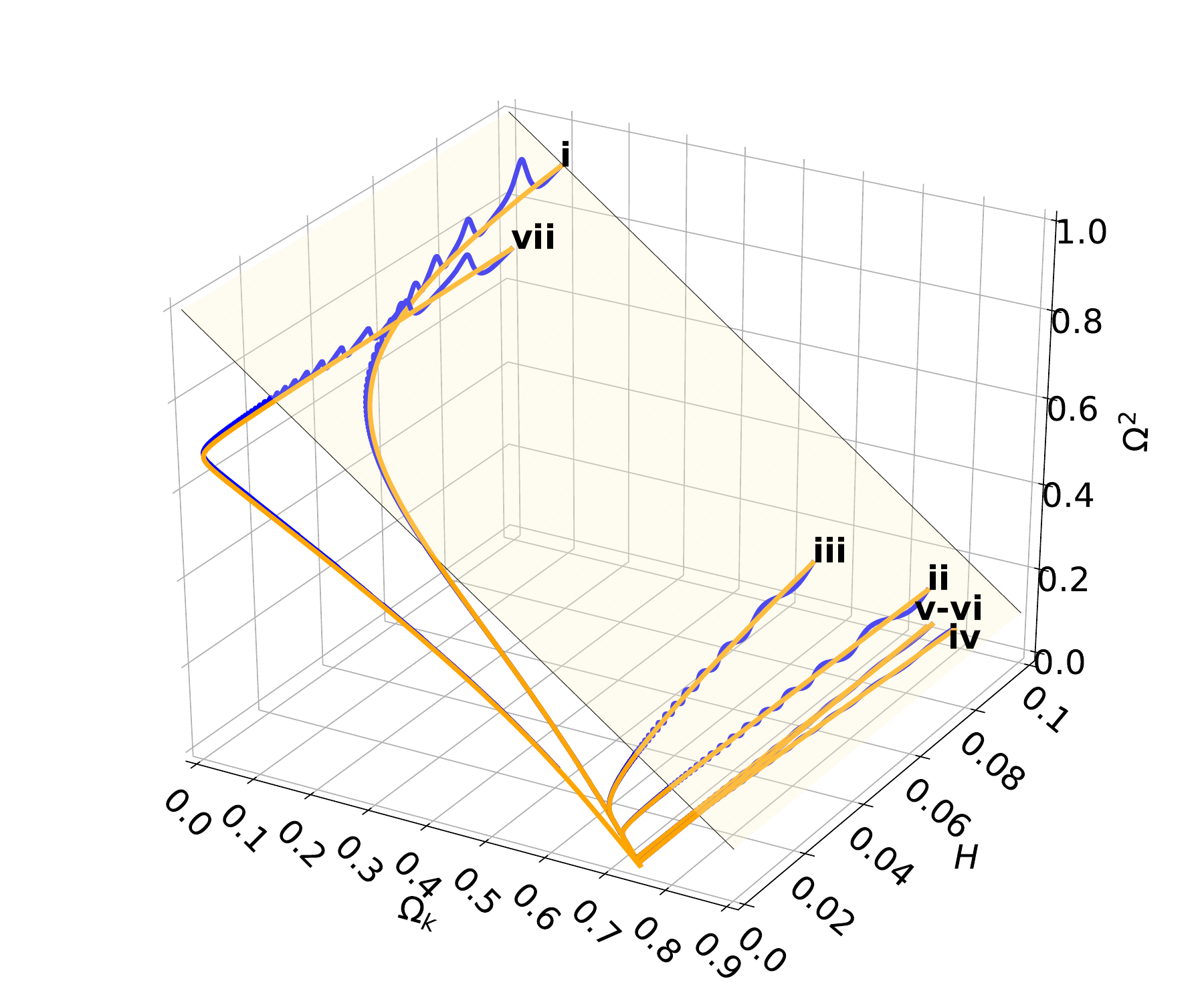}
    \caption{LRS Bianchi III with $\gamma=1$: projection in $(\Omega_k,H,\Omega^2)$.}
    \label{fig:BI3_K}
  \end{subfigure}\hfill
  \begin{subfigure}[b]{0.30\textwidth}
    \centering
    \includegraphics[width=\textwidth]{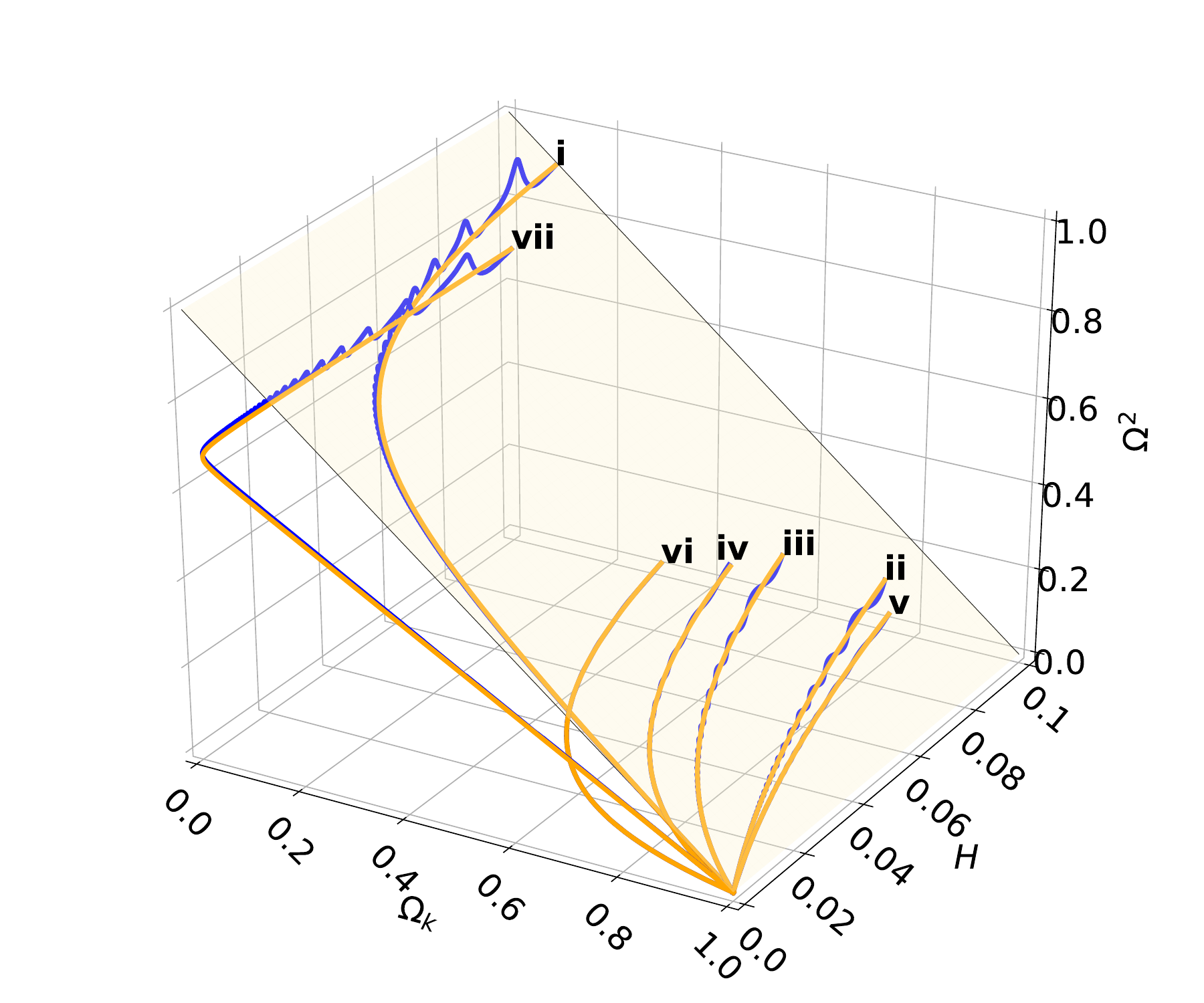}
    \caption{Open FLRW ($k=-1$): projection in $(\Omega_k,H,\Omega^2)$.}
    \label{fig:OpenFLRW}
  \end{subfigure}

  \caption{\label{fig:curvedgeometry3D} Three-dimensional phase-space projections for curved geometries. LRS Bianchi III highlights anisotropic shear and curvature effects; open FLRW illustrates isotropic evolution with negative spatial curvature.}
\end{figure}

In figure \ref{fig:curvedgeometry3D}, three-dimensional phase space projections for curved geometries are represented. LRS Bianchi III highlights anisotropic shear and curvature effects; open FLRW illustrates isotropic evolution with negative spatial curvature.

\section{Conclusions}
\label{sect:4}

We develop a unified, easy-to-apply method that combines dynamical-systems ideas with averaging to study oscillating scalar-field cosmologies with quasi-harmonic potentials and matter coupling, including curvature and anisotropy. Under the regularity, periodicity and averaging decomposition assumptions \ref{A1}--\ref{A3} (which give $\mathbf f^1=\bar{\mathbf f}^1+\tilde{\mathbf f}^1$ with $\tilde{\mathbf f}^1$ of zero $\theta$–mean and control the remainders), the averaged‑in‑$\theta$ system closely follows the full oscillatory system: the two flows are related by a smooth near‑identity change of variables and the pointwise error satisfies $\|\mathbf x(t)-\bar{\mathbf x}(t)\|=\mathcal O(H(t))$ as $H(t)\to0$. The uniform error bound and the conjugacy use the common small initial data and Lipschitz bounds in \ref{A4} and \ref{A3}, convergence of the averaged solution to an equilibrium uses \ref{A5}, and the harmlessness of boundary terms from integration by parts is ensured by the frequency–scaling condition \ref{A6} (or by the model dispersion relation $b\mu^3+2f\mu^2-f\omega^2=0$, leading the special case $\mathbf f^0\equiv0$). As a result, equilibria and invariant sets of the averaged system give reliable late‑time scaling laws and effective equations of state, which lead to the compact, geometry‑dependent attractor classification stated above (for example, flat FLRW and LRS Bianchi I select $\mathcal{F}_0$ when $0<\gamma<1$ and $\mathcal{F}$ when $1<\gamma\le2$; open FLRW selects $\mathcal{F}_0$ for $0<\gamma\le\tfrac{2}{3}$ and Milne $\mathcal{C}$ for $\tfrac{2}{3}<\gamma<2$; analogous ranges hold for LRS Bianchi III and V as listed). Closed FLRW and Kantowski–Sachs require replacing $H$ by the curvature‑sensitive normalization $D=\sqrt{H^2+\tfrac{1}{6}\,{}^{(3)}\!R}$ and the corresponding versions of \ref{A1}--\ref{A3}; there the averaged description is valid only when $D(t)\to0$ uniformly, otherwise a tailored $D$‑based analysis is needed \cite{Leon:2021hxc}. Finally, the same Hubble‑based averaging extends to interacting scalar–matter systems in anisotropic spacetimes (excluding closed geometries), so under the explicitly cited assumptions \ref{A1}--\ref{A6} the framework provides a practical, rigorous procedure to identify attractors, locate bifurcations, and quantify approximation errors in oscillatory scalar‑field cosmologies.

In summary, this unified averaging framework isolates the essential slow dynamics hidden beneath rapid scalar‑field oscillations and renders them analytically tractable. It turns a highly oscillatory, geometry‑dependent system into a controlled perturbation of an autonomous averaged flow with explicit error bounds. This provides a robust foundation for systematic late‑time analysis across a wide class of scalar‑field cosmologies, from flat FLRW to anisotropic open models.

\acknowledgments{Funded by Agencia Nacional de Investigación y Desarrollo (ANID), Chile, through Proyecto Fondecyt Regular 2024, Folio 1240514, Etapa 2025.}

\bibliographystyle{unsrt}
\bibliography{refs.bib}

\end{document}